\documentclass[11pt]{article}

\usepackage[letterpaper,margin=1in]{geometry}
\usepackage[T1]{fontenc}
\usepackage[utf8]{inputenc}
\usepackage{lmodern}
\usepackage{microtype}
\usepackage{xcolor}
\usepackage{amsmath,amssymb,amsthm}
\usepackage{enumitem}
\usepackage[authoryear,round]{natbib}
\usepackage[colorlinks=true,linkcolor=blue,citecolor=blue,urlcolor=blue]{hyperref}
\AtBeginDocument{\lefthyphenmin=2\righthyphenmin=3}

\newcommand{\N}{\mathbb N}
\newcommand{\Np}{\mathbb N_{\geq 1}}
\newcommand{\hashPone}{\mathsf{\#P}_1}
\newcommand{\FP}{\mathsf{FP}}
\newcommand{\poly}{\operatorname{poly}}
\newcommand{\bin}{\operatorname{bin}}
\newcommand{\rad}{\operatorname{rad}}
\definecolor{revision}{RGB}{0,0,0}
\definecolor{revisiontwo}{RGB}{0,0,0}
\definecolor{revisionthree}{RGB}{0,0,0}
\definecolor{revisionfour}{RGB}{0,0,0}
\definecolor{revisionfive}{RGB}{0,0,0}
\definecolor{revisionsix}{RGB}{0,0,0}
\definecolor{revisionseven}{RGB}{0,0,0}
\definecolor{revisioneight}{RGB}{0,0,0}
\definecolor{revisionnine}{RGB}{0,0,0}
\definecolor{revisionten}{RGB}{0,0,0}
\definecolor{revisioneleven}{RGB}{0,0,0}
\definecolor{revisiontwelve}{RGB}{0,0,0}
\definecolor{revisionthirteen}{RGB}{0,0,0}
\definecolor{revisionfourteen}{RGB}{0,0,0}
\definecolor{revisionfifteen}{RGB}{0,0,0}
\definecolor{revisionsixteen}{RGB}{0,0,0}
\definecolor{revisionseventeen}{RGB}{0,0,0}
\definecolor{revisioneighteen}{RGB}{0,0,0}
\definecolor{revisionnineteen}{RGB}{0,0,0}
\definecolor{revisiontwenty}{RGB}{0,0,0}
\definecolor{revisiontwentyone}{RGB}{0,0,0}
\newcommand{\rev}[1]{\textcolor{revision}{#1}}

\newcommand{\revprec}[1]{\textcolor{revisionfour}{#1}}
\newcommand{\revreduce}[1]{\textcolor{revisionfive}{#1}}

\newcommand{\revtodo}[1]{\textcolor{revisionthirteen}{#1}}
\newcommand{\revfinal}[1]{\textcolor{revisionfourteen}{#1}}

\definecolor{currentrevision}{RGB}{0,0,0}

\definecolor{revisiontwentytwo}{RGB}{0,0,0}

\definecolor{revisiontwentythree}{RGB}{0,0,0}
\newcommand{\revreadability}[1]{\textcolor{revisiontwentythree}{#1}}
\definecolor{revisiontwentyfour}{RGB}{0,0,0}

\definecolor{revisiontwentyfive}{RGB}{0,0,0}
\newcommand{\revglobal}[1]{\textcolor{revisiontwentyfive}{#1}}
\definecolor{revisiontwentysix}{RGB}{0,0,0}
\newcommand{\revbridge}[1]{\textcolor{revisiontwentysix}{#1}}
\definecolor{revisiontwentyseven}{RGB}{0,0,0}

\definecolor{revisiontwentyeight}{RGB}{0,0,0}

\definecolor{revisionthirty}{RGB}{0,0,0}

\definecolor{revisionthirtyone}{RGB}{0,0,0}

\definecolor{revisionthirtytwo}{RGB}{0,0,0}
\newcommand{\revframing}[1]{\textcolor{revisionthirtytwo}{#1}}
\definecolor{revisionthirtythree}{RGB}{0,0,0}

\definecolor{revisionthirtyfour}{RGB}{0,0,0}

\definecolor{revisionthirtyfive}{RGB}{0,0,0}

\definecolor{revisionthirtysix}{RGB}{0,0,0}

\definecolor{revisionthirtyseven}{RGB}{0,0,0}
\newcommand{\revseam}[1]{\textcolor{revisionthirtyseven}{#1}}
\definecolor{revisionthirtyeight}{RGB}{0,0,0}
\newcommand{\revstyle}[1]{\textcolor{revisionthirtyeight}{#1}}
\definecolor{revisionthirtynine}{RGB}{0,0,0}
\newcommand{\revrepair}[1]{\textcolor{revisionthirtynine}{#1}}
\definecolor{revisionforty}{RGB}{0,0,0}
\newcommand{\revproof}[1]{\textcolor{revisionforty}{#1}}
\definecolor{revisionfortyone}{RGB}{0,0,0}

\definecolor{revisionfortytwo}{RGB}{0,0,0}

\definecolor{revisionfortythree}{RGB}{0,0,0}

\definecolor{revisionfortyfour}{RGB}{0,0,0}
\newcommand{\revflow}[1]{\textcolor{revisionfortyfour}{#1}}
\definecolor{revisionfortysix}{RGB}{0,0,0}
\newcommand{\revmerged}[1]{\textcolor{revisionfortysix}{#1}}
\definecolor{revisionfortyseven}{RGB}{0,0,0}

\definecolor{revisionfortyeight}{RGB}{0,0,0}
\newcommand{\revreferee}[1]{\textcolor{revisionfortyeight}{#1}}
\newtheorem{theorem}{Theorem}
\newtheorem{lemma}[theorem]{Lemma}
\newtheorem{corollary}[theorem]{Corollary}

\hypersetup{
  pdftitle={Accepting-Path Counting at the One-Tape n log n Threshold},
  pdfauthor={Ondrej Kuzelka},
  pdfsubject={Accepting-path counting on one-tape Turing machines}
}

\title{Accepting-Path Counting at the One-Tape $n\log n$ Threshold}
\author{Ond\v{r}ej Ku\v{z}elka\\
Faculty of Electrical Engineering, Czech Technical University in Prague\\
Prague, Czech Republic}
\date{}

\begin{document}
\maketitle

\begin{abstract}
{\color{revisionthirtysix}
We observe that the classical $n\log n$ time threshold for one-tape Turing machines is also a threshold for their accepting-path counts. Below it, every nondeterministic one-tape machine running in strong $o(n\log n)$ time has a rational ordinary generating function of accepting-path counts. At strong $O(n\log n)$ time, the situation changes completely: there is a fixed one-tape machine whose accepting-path function is complete for $\hashPone$, the tally analogue of $\#\mathsf P$, under parsimonious polynomial-time tally reductions. A second construction within the same time bound gives positive accepting-path counts with a noncomputable exponential growth rate. \revrepair{The rationality result combines the one-tape time gap with the linear-time counting theorem of Tadaki, Yamakami and Lin. \revmerged{The completeness proof adapts the linear-time universal counting machine of Beame et al. to the one-tape setting.}}
}
\end{abstract}

\section{Introduction}

{\color{revisionthree}
Valiant's class $\#\mathsf P$ captures the complexity of counting accepting computations of polynomial-time nondeterministic machines \citep{Valiant1979}.
Its tally analogue, \revtodo{denoted by $\hashPone$}, asks the same question on unary inputs: a function $f:\Np\to\N$ belongs to $\hashPone$ if some nondeterministic polynomial-time machine has exactly $f(n)$ accepting paths on input $1^n$.
We ask how much time is needed before a machine restricted to one read/write tape and one head can have an accepting-path counting function that is \revprec{$\hashPone$-complete under a parsimonious polynomial-time tally reduction}.

\revstyle{This question is motivated by other settings in which a fixed object gives rise to a $\hashPone$-complete counting sequence.}
\revbridge{For example, under polynomial-time oracle reductions, even fixed languages recognized by surprisingly weak devices can give rise to $\hashPone$-complete sequences counting the number of words of each length that belong to the language \citep{BertoniGoldwurm1993}.}
Likewise, first-order model counting fixes a \revseam{sentence of first-order logic} and counts its \revtodo{models} of each finite size; Beame, Van den Broeck, Gribkoff and Suciu exhibited a fixed three-variable sentence with a $\hashPone$-complete model-counting sequence \revreduce{under the same type of reduction} \citeyearpar{BeameEtAl2015}.
\revmerged{Beame et al. also show that, when several tapes are available, $\hashPone$-hard accepting-path counting is possible already in linear time: they construct a fixed linear-time counting machine whose accepting-path function is $\hashPone$-hard \citeyearpar{BeameEtAl2015}.}

{\color{revisionthirtysix}
Our answer to the question of when $\hashPone$-complete accepting-path counting becomes possible on a one-tape machine is as follows. Below $n\log n$, strong $o(n\log n)$ time forces the accepting-path generating function to be rational. Such a fixed rational sequence is polynomial-time computable, so a $\hashPone$-hard accepting-path function below the threshold under our reductions would imply $\hashPone\subseteq\FP$. \revreferee{Thus, under the standard assumption $\hashPone\nsubseteq\FP$, hardness cannot occur below the threshold.} At strong $O(n\log n)$ time, the opposite extreme is already possible: there is a fixed one-tape machine with a $\hashPone$-complete accepting-path function. \revrepair{The lower boundary is the classical one-tape time gap, while the upper result shows that the same threshold already suffices for $\hashPone$-complete accepting-path counting.}
}

\revbridge{For a fixed nondeterministic one-tape machine $M$, let $a_M(n)$ denote its number of accepting paths on $1^n$, and define the corresponding ordinary generating function by}
\[
 A_M(z)=\sum_{n\geq0}a_M(n)z^n.
\]
\revstyle{The main result is the following.}
}

\begin{theorem}\label{thm:main}
\revglobal{For nondeterministic one-tape machines, with running time measured as the maximum over all computation paths, the following hold.}
\begin{enumerate}[label=\textup{(\roman*)},leftmargin=*,itemsep=3pt]
\item\label{item:rational}
If $M$ runs in $o(n\log n)$ time, then $A_M(z)$ is rational over $\mathbb Q$.
\item\label{item:universal}
{\color{revisionthirtysix}
There is a fixed one-tape machine $U$ running in strong $O(n\log n)$ time such that, for every $f\in\hashPone$, there is a polynomially bounded, polynomial-time computable map $p_f:\Np\to\Np$ satisfying
\[
 f(j)=a_U(p_f(j))\qquad(j\geq1).
\]
Thus $a_U$ is $\hashPone$-complete under parsimonious polynomial-time tally reductions.
}
\item\label{item:radius}
There is a fixed machine $R$ running in \revreferee{strong $O(n\log n)$ time} such that $a_R(n)>0$ and the limit $\lim_{n\to\infty}a_R(n)^{1/n}$ exists and is a noncomputable real.
Consequently, $A_R(z)$ has a noncomputable radius of convergence.
\end{enumerate}
\end{theorem}

{\color{revisionthirtysix}
\revmerged{The tools are known. This note combines them to place a sharp $n\log n$ threshold on one-tape accepting-path counting, and supplies a one-tape implementation of Beame et al.'s construction that yields the completeness half of that threshold.} Part~\ref{item:rational} follows directly by combining the one-tape time gap with the linear-time counting theorem of Tadaki, Yamakami and Lin. \revmerged{For Part~\ref{item:universal}, the relevant ingredient is the universal counting construction of Beame et al.~\citeyearpar[Lemma~3.8]{BeameEtAl2015}, which encodes the accepting-path counts of arbitrary unary polynomial-time machines into those of one fixed linear-time machine; Section~\ref{sec:universal} gives the one-tape implementation while preserving these counts exactly.} Part~\ref{item:radius} is an independent elementary construction showing that the same time bound permits noncomputable asymptotic growth.
}

{\color{revisionten}
\revframing{This paper is organized as follows.} Section~\ref{sec:background} reviews the background needed for these statements. Section~\ref{sec:rational} derives the rationality consequence, Section~\ref{sec:universal} gives the $\hashPone$-complete one-tape construction, and Section~\ref{sec:radius} gives the noncomputable-growth construction.
}

\section{Background and preliminaries}
\label{sec:background}

\subsection{Tally counting and reductions}

{\color{revisionthree}
Throughout, $\hashPone$ denotes tally $\#\mathsf P$ with the exact unary input length as parameter.
Thus $f\in\hashPone$ if a nondeterministic machine has exactly $f(n)$ accepting paths on $1^n$ and every path runs in time polynomial in $n$.
The machine may use any fixed number of tapes. \revmerged{Equivalently, one may restrict to a read-only input tape and one work tape: the standard tape simulation can be carried out deterministically between nondeterministic choices, with polynomial slowdown and with accepting paths preserved one for one.}

\revfinal{For any such machine $M$, let $T_M(n)$ denote the maximum running time over all computation paths of $M$ on $1^n$. This is the {\em strong time measure} used below.}

We use {\em parsimonious polynomial-time tally reductions}.
A reduction from $f$ to $g$ maps $1^j$ to $1^{p(j)}$ in polynomial time and satisfies $f(j)=g(p(j))$. \revmerged{Equivalently, $p(j)$ is polynomially bounded and its binary value is computable in time polynomial in $j$.}
In our construction the binary value of $p(j)$ is in fact computable in time polynomial in $\log(j+2)$; \revfinal{producing the unary output $1^{p(j)}$ then takes time polynomial in $j$.}

We use ordinary generating functions to encode the counting sequences considered below. For a power series $F(z)=\sum_{n\geq0}c_nz^n$, we denote its radius of convergence by $\rad(F)$; the series converges absolutely for $|z|<\rad(F)$ and diverges for $|z|>\rad(F)$. In Section~\ref{sec:radius} we use the Cauchy--Hadamard formula
\[
 \frac{1}{\rad(F)}=\limsup_{n\to\infty}|c_n|^{1/n},
\]
which expresses the radius directly in terms of the coefficients.

\subsection{\texorpdfstring{\revfinal{Existing $\hashPone$-complete counting results}}{Existing \#P1-complete counting results}}
\label{sec:beame}

\revmerged{The closest precursor of our completeness result is Lemma~3.8 of Beame et al.~\citeyearpar{BeameEtAl2015}. They construct a fixed linear-time counting machine $U_1$ on unary inputs whose accepting-path function is $\hashPone$-hard. Their proof enumerates clocked counting machines and encodes the machine index together with its source input length into the unary input length of $U_1$. Although hardness is stated there under polynomial-time oracle reductions, the proof gives the stronger exact correspondence needed here: for every $g\in\hashPone$, one fixed machine index $i$ satisfies $g(j)=a_{U_1}(e(i,j))$ at all encoded lengths, for the explicit encoding $e$ recalled in Section~\ref{sec:universal}. We revisit that proof directly in Section~\ref{sec:universal}, retaining their clocked enumeration and their encoding rather than generically simulating the completed machine $U_1$.}

Beame et al. subsequently use their universal-counter construction to obtain a $\hashPone$-hard first-order model-counting result for a fixed sentence with three variables. We use only the machine construction, not the later encoding into first-order models.

Another relevant source of $\hashPone$-hard counting sequences is the work of Bertoni and Goldwurm \citeyearpar{BertoniGoldwurm1993}. For a fixed language $L\subseteq\Sigma^*$, they consider the exact-length counting function
\[
 c_L(n)=|L\cap\Sigma^n|.
\]
\revmerged{They show that even very weak language classes, including languages recognized by one-way deterministic two-head finite automata, contain a language with a $\hashPone$-complete exact-length counting function under one-query polynomial-time Turing reductions. In their Proposition~8, accepting computations of a source polynomial-time counting machine are encoded by words of a prescribed length in a language recognized by such a device, giving an exact equality between the source count and the number of accepted words at that length. Choosing a fixed $\hashPone$-complete source function therefore yields a fixed language with a $\hashPone$-complete counting sequence. This is closely related in spirit to our construction, but it concerns a different counting object and a different reduction notion.}

\subsection{\texorpdfstring{One-tape machines and the $n\log n$ boundary}{One-tape machines and the n log n boundary}}

{\color{revisionthirtysix}By a {\em one-tape machine} we mean a Turing machine with a single two-way infinite read/write tape, initially containing the input, which is used for both input and work. The head starts at the first input cell. \revseam{One-tape input/work models of this kind are often called ``off-line'' Turing machines in the one-tape literature \citep{Pighizzini2009}.}}
\revreadability{As usual, we may encode a fixed amount of auxiliary information in each tape cell by enlarging the tape alphabet. We allow the tape head to remain in place during a step, and we use designated accepting and rejecting halting states. These conventions do not affect the results.\footnote{A stay-put move can be replaced by a fixed sequence of ordinary moves, with only constant overhead and without changing the number of accepting computation paths.}}

\revreadability{We use two established facts about this model. \revflow{First, Gajser's one-tape time-gap theorem states that if a nondeterministic one-tape machine runs in strong $o(n\log n)$ time, then it in fact runs in $O(n)$ time} \citeyearpar[Corollary~5.1.6]{GajserThesis}. The proof uses crossing-sequence techniques; we refer to \citet{GajserThesis} for details.}

\revglobal{Second, Tadaki, Yamakami and Lin show that the accepting-path counting function of a nondeterministic one-tape machine running in strong linear time has a finite-dimensional linear representation \citeyearpar[proof of Lemma~7.3]{TadakiYamakamiLin2010}. \revreferee{In their counting construction, the entries record numbers of local computation paths, so the vectors and matrices can be taken to have nonnegative integer entries.} More precisely, there are fixed vectors $\pi,\eta$ and, for each input symbol $\sigma$, a fixed matrix $T(\sigma)$ such that the number of accepting computation paths on $x=x_1\cdots x_n$ is
\[
 \pi T(x_1)\cdots T(x_n)\eta.
\]
For unary inputs there is only one transition matrix. Writing $A=T(1)$, we therefore have
\[
 a_M(n)=\pi A^n\eta.
\]}
}

\section{Rationality below \texorpdfstring{$n\log n$}{n log n}}
\label{sec:rational}

{\color{revisionthree}
\revglobal{The two ingredients needed here were stated in Section~\ref{sec:background}: the strong $o(n\log n)$-to-$O(n)$ runtime gap and the \revreferee{linear representation} of strong linear-time one-tape counting.}
}

\begin{proof}[Proof of Theorem~\ref{thm:main}\ref{item:rational}]
\revglobal{The runtime gap puts $M$ in strong linear time. By the linear representation recalled in Section~\ref{sec:background}, there are fixed \revreferee{nonnegative integer} vectors $\pi,\eta$ and a fixed \revreferee{nonnegative integer} matrix $A$ such that $a_M(n)=\pi A^n\eta$. Therefore
\[
 A_M(z)=\sum_{n\geq0}a_M(n)z^n
 =\pi\left(\sum_{n\geq0}(zA)^n\right)\eta
 =\pi(I-zA)^{-1}\eta.
\]
Since $A$ is a finite \revreferee{integer} matrix, every entry of $(I-zA)^{-1}$ is a rational function of $z$, and hence so is $A_M(z)$.}
\end{proof}

\begin{corollary}\label{cor:rational}
If $T_M(n)=o(n\log n)$, then $a_M$ satisfies a constant-coefficient linear recurrence and is computable in time polynomial in the unary input length.
Moreover, $\rad(A_M)$ is either infinite or a positive computable algebraic number.
\end{corollary}

\begin{proof}
The recurrence follows from the rational generating function.
\revproof{\revreferee{For the complexity claim, compute $\pi A^n\eta$ by repeated multiplication by the fixed integer matrix $A$. The dimension and entries are fixed, so all intermediate integers have $O(n)$ bits. This gives an exact polynomial-time algorithm.}}
Since $A_M$ is rational, there exist coprime polynomials $P,Q\in\mathbb Q[z]$ such that
\[
 A_M(z)=\frac{P(z)}{Q(z)}
\]
and $Q(0)\neq0$.
If $Q$ is constant, the radius is infinite.
Otherwise, the radius is the minimum modulus of a zero of $Q$, and is therefore a positive computable algebraic real.
\end{proof}

\rev{Thus a $\hashPone$-hard counting function below the threshold, under polynomial-time oracle reductions, would make every $\hashPone$ function computable in deterministic polynomial time on its unary input.}
The rationality and radius statements themselves require no complexity separation assumption.

\section{Counting at \texorpdfstring{$O(n\log n)$}{O(n log n)} time}
\label{sec:universal}

\revmerged{A generic path-preserving simulation of the linear-time machine $U_1$ constructed by Beame et al.~\citeyearpar[Lemma~3.8]{BeameEtAl2015} would give an $O(N^2)$-time one-tape machine on inputs of length $N$; this is the route used by Kuang et al.~\citeyearpar[proof of Lemma~3.2]{KuangEtAl2025}. Instead, we follow the proof of Beame et al. before the final universal machine is assembled. Their construction already supplies the clocked machines and an encoding with enough slack to pay for a quadratic one-tape simulation of the selected clocked machine. As we show in this section, the only remaining step whose cost increases asymptotically in our model is the initial conversion of the unary input length to binary.}

\begin{lemma}\label{lem:binary}
A deterministic one-tape machine can compute $\bin(n)$ from $1^n$ in $O(n\log n)$ time.
It can finish with this binary word in a compact work area and the remaining tape blank.
\end{lemma}

\begin{proof}
We mark the ends of the original input interval and initially regard every input cell as active.
In one $O(n)$-time round, the machine retains every second active mark and records the parity of the number of active marks in the finite control.
If there were $m$ active marks, this leaves $\lfloor m/2\rfloor$ active marks and produces the next binary digit $m\bmod2$, which is appended to a compact output area.
After $O(\log n)$ rounds no active marks remain.
Each round costs $O(n+\log n)$, including access to the output area.
A final sweep erases the input and scratch information.
Reversing the $O(\log n)$ output bits, if needed, costs $O((\log n)^2)$.
\end{proof}

\begin{proof}[Proof of Theorem~\ref{thm:main}\ref{item:universal}]
\revmerged{We follow the proof of Beame et al.~\citeyearpar[Lemma~3.8]{BeameEtAl2015}, recalling only the ingredients needed for the one-tape implementation. Their construction starts with a standard enumeration
\[
 P_1,P_2,\ldots
\]
of unary counting machines. It then considers all pairs $(P_r,s)$, where $s\geq1$ specifies a polynomial clock, and effectively enumerates these pairs by indices $i$ chosen so that $r,s\leq i$. On input $1^j$, the clocked machine associated with $(P_r,s)$ simulates $P_r$ for at most $sj^s+s$ steps, rejecting a computation path if it has not halted by then. This still captures every function in $\hashPone$: for any polynomial-time counting machine $P_r$, one can choose $s$ large enough that the cutoff never interrupts any of its computation paths. Maintaining the clock incurs at most a quadratic slowdown. Thus Beame et al. obtain clocked machines
\[
 M_1,M_2,\ldots
\]
such that every function in $\hashPone$ is computed by some $M_i$ and every computation path of $M_i$ on input $1^j$ has length at most
\begin{equation}\label{eq:beame-clocked-bound}
 T_i(j)=(ij^i+i)^2.
\end{equation}}

\revmerged{Beame et al. then combine these clocked machines into one fixed universal counting machine $U_1$. Its unary input length specifies both which machine $M_i$ is to be simulated and the input length $j$ for that machine. To encode this information, they use
\begin{equation}\label{eq:encoding}
 e(i,j)=2^i3^{4i\lceil\log_3j\rceil}(6j+1).
\end{equation}
On input $1^{e(i,j)}$, the machine $U_1$ recovers $i$ and $j$ and simulates $M_i$ on $1^j$. Consequently, for every $g\in\hashPone$, some fixed index $i$ satisfies
\[
 g(j)=a_{U_1}(e(i,j))\qquad(j\geq1).
\]
Lemma~3.8 of \citet{BeameEtAl2015} also establishes that $i$ and $j$ can be recovered once the input length is available in binary, and that, for each fixed $i$, $j\mapsto e(i,j)$ is polynomially bounded and polynomial-time computable. We shall also use the immediate estimate
\begin{equation}\label{eq:encoding-slack}
 e(i,j)\geq 2^i j^{4i}(6j+1),
\end{equation}
which follows from $3^{\lceil\log_3j\rceil}\geq j$.}

\revmerged{We now implement this encoded simulation by a fixed one-tape machine $U$. On input $1^N$, it first computes $\bin(N)$ by Lemma~\ref{lem:binary}. \revmerged{It applies the decoding procedure of Beame et al. to obtain a candidate pair $(i,j)$ and verifies that $N=e(i,j)$; if the verification fails, it rejects deterministically.} Once $\bin(N)$ is available, this arithmetic uses only numbers of $O(\log N)$ bits, so any fixed polynomial-time implementation costs $O(N)$ time. On a valid input it recovers the pair $(P_r,s)$ indexed by $i$. Using a standard binary encoding, the description of $P_r$ can be recovered in $\poly(\log(i+2))$ time, since $r\leq i$, while the input $1^j$ can be prepared from the binary representation of $j$ in $O(j^2)$ time. Both costs are $O(N)$, since $N\geq2^i$ and, by \eqref{eq:encoding-slack}, $N\geq 2j^4(6j+1)$.}

\revmerged{It remains to simulate $M_i$. A standard path-preserving one-tape simulation of its read-only input tape and work tape takes quadratic time in the running time of $M_i$, with an additional overhead polynomial in the size of its description. Thus there are a constant $A>0$ and a fixed polynomial $q$, independent of $i$ and $j$, such that the simulation time is at most $A T_i(j)^2 q(\log(i+2))$. Since $j^i\geq1$, we have $T_i(j)^2=(ij^i+i)^4\leq16i^4j^{4i}$. Because $q$ is fixed, there is a constant $B$ such that $16A i^4q(\log(i+2))\leq B2^i$ for every $i\geq1$. Hence, using \eqref{eq:encoding-slack},
\[
 A T_i(j)^2q(\log(i+2))
 \leq B2^i j^{4i}
 \leq B e(i,j)=BN.
\]
Thus the simulation takes $O(N)$ time with a constant independent of $i$ and $j$. It uses deterministic bookkeeping except when it reproduces a nondeterministic transition of $M_i$, so it preserves computation paths and their accepting or rejecting outcomes one for one.}

\revmerged{Now fix $f\in\hashPone$. By the enumeration of Beame et al., some fixed $M_i$ computes $f$, and the universal machine $U_1$ therefore satisfies $f(j)=a_{U_1}(e(i,j))$ for every $j\geq1$. Our one-tape simulation preserves the same computation paths, so
\[
 f(j)=a_U(e(i,j)).
\]
Set $p_f(j)=e(i,j)$. The index $i$ depends only on $f$ and is therefore constant with respect to $j$. Moreover,
\[
 p_f(j)=2^i3^{4i\lceil\log_3 j\rceil}(6j+1)
 \leq 2^i3^{4i}j^{4i}(6j+1),
\]
so $p_f$ is polynomially bounded in $j$. Its binary value is computable in time polynomial in $\log(j+2)$. Producing the unary output $1^{p_f(j)}$ therefore takes polynomial time in $j$, so $p_f$ is a parsimonious polynomial-time tally reduction. Finally, the initial unary-to-binary conversion costs $O(N\log N)$ and every subsequent stage costs $O(N)$; hence $U$ runs in strong $O(N\log N)$ time. \revmerged{In particular $a_U\in\hashPone$; together with the parsimonious hardness above, this proves $\hashPone$-completeness.}}
\end{proof}

\section{A noncomputable exponential growth rate}
\label{sec:radius}

A real is computable if an algorithm can, on input $m$, produce a rational approximation with error at most $2^{-m}$.
A computable sequence of coefficients can have a noncomputable radius of convergence.
For example, Gra\c{c}a, Zhong and Buescu \citeyearpar[Lemma~4.3]{GracaZhongBuescu2009} use such a power series in their study of computable analytic differential equations.
We construct positive integer coefficients with this property that are accepting-path counts under the one-tape time bound of Theorem~\ref{thm:main}.
The construction uses textbook computability-theoretic ingredients, essentially the classical Specker phenomenon that a computable increasing sequence of rationals can have a noncomputable limit \citep{Specker1949}; the point here is that they can be implemented within the present one-tape time bound.

\revrepair{Our construction of a sequence with noncomputable exponential growth separates the computation into two parts. On input $1^n$, the machine first performs a deterministic computation that produces an integer $k(n)$. It then makes exactly $k(n)$ binary nondeterministic choices and accepts on every resulting path. Its number of accepting paths is therefore}
\[
 a_n=2^{k(n)}.
\]
\revrepair{We therefore seek $k(n)$ such that $k(n)/n\to\alpha$ for a noncomputable real $\alpha$; this will give $a_n^{1/n}\to2^\alpha$.}

\revrepair{The machine cannot use $\alpha$ directly in its deterministic preprocessing. Instead, it uses computable finite approximations $\alpha_s\to\alpha$. This does not make $\alpha$ computable: although each $\alpha_s$ is computable, there is no computable bound on how large $s$ must be to achieve a prescribed accuracy. On input $1^n$, the machine chooses}
\[
 s(n)=\lfloor\log_2 n\rfloor
\]
\revrepair{and sets $k(n)=\lfloor n\alpha_{s(n)}\rfloor$. We define $\alpha$ and the approximations $\alpha_s$ explicitly in the proof below.}

\begin{proof}[Proof of Theorem~\ref{thm:main}\ref{item:radius}]
\revproof{We fix a binary transition-table encoding of deterministic Turing machines.}
\revproof{Enumerate all finite binary strings as $w_0,w_1,\ldots$ in length-lexicographic order. Given the index $e$, the string $w_e$ can be recovered in time polynomial in $\log(e+2)$. Let $Q_e$ be the machine described by $w_e$, with malformed descriptions denoting a fixed nonhalting machine.}
Consequently
\[
 K=\{e:Q_e\text{ halts on the empty input}\}
\]
is undecidable.
\revproof{We define}
\begin{equation}\label{eq:alpha}
 \alpha=\frac14+\sum_{e\in K}4^{-(e+2)},
 \qquad \frac14\leq\alpha\leq\frac13.
\end{equation}
This real is noncomputable.
\revproof{Indeed, suppose we have determined membership in $K$ for indices below $e$. We subtract their contributions and $1/4$ from $\alpha$, and multiply by $4^{e+2}$.}
The result is
\[
 \mathbf1_K(e)+\sum_{r>e}\mathbf1_K(r)4^{e-r}.
\]
Consequently, the displayed quantity is at most $1/3$ when $e\notin K$, and at least $1$ when $e\in K$.
An algorithm approximating $\alpha$ to arbitrary prescribed accuracy would decide these cases successively and hence decide $K$.

\revproof{For each $s\geq0$, we define the approximation $\alpha_s$ by retaining only computations that halt within $s$ steps:}
\begin{equation}\label{eq:stages}
 \alpha_s=\frac14+
 \sum_{\substack{e\leq s\\Q_e\text{ halts within }s\text{ steps}}}4^{-(e+2)}.
\end{equation}
Then $\alpha_s$ increases to $\alpha$.
We can compute $\alpha_s$ exactly in time polynomial in $s$: simulate the finitely many machines $Q_0,\ldots,Q_s$ for $s$ steps, then add the indicated fractions.
The descriptions have $O(\log(s+2))$ bits, and the common denominator $4^{s+2}$ and the numerator have $O(s+1)$ bits.
A fixed one-tape implementation of these finite simulations and arithmetic operations still takes polynomial time in $s$.

\revproof{For input length $n$, we choose $s(n)=\lfloor\log_2 n\rfloor$. Then $s(n)\to\infty$, so $\alpha_{s(n)}\to\alpha$, while computing $\alpha_{s(n)}$ takes only $\poly(\log n)$ time.}
We now realize the sequence described above by a single one-tape machine $R$. On the empty input, $R$ accepts immediately.
\revproof{On input $1^n$ with $n\geq1$, the machine first computes $\bin(n)$ and clears the scratch area using Lemma~\ref{lem:binary}. It then computes
\[
 s(n)=\lfloor\log_2 n\rfloor,
 \qquad
 k(n)=\lfloor n\alpha_{s(n)}\rfloor,
\]
including the exact rational $\alpha_{s(n)}$ needed to obtain $k(n)$.}
This further work takes $\poly(\log n)$ time.
\revproof{Finally, it makes exactly $k(n)$ binary choices and accepts on every resulting path.}
Each choice uses two distinct successor states that deterministically return to the same countdown procedure.
There are exactly $2^{k(n)}$ accepting paths.
Since $k(n)\leq n/3$, maintaining a binary counter for the remaining choices takes $O(\log n)$ time per choice, and hence $O(n\log n)$ time altogether.

It remains to compute the growth rate.
For every $n\geq1$,
\[
 0\leq\alpha_{s(n)}-\frac{k(n)}n<\frac1n.
\]
As $s(n)\to\infty$, it follows that $k(n)/n\to\alpha$, and therefore
\[
 \lim_{n\to\infty}a_n^{1/n}=2^\alpha,
 \qquad
 \rad\!\left(\sum_{n\geq0}a_nz^n\right)=2^{-\alpha}
\]
by the Cauchy--Hadamard formula.
Both reals are noncomputable: computability of either one, followed by the computable logarithm on a positive interval bounded away from zero, would make $\alpha$ computable.
\end{proof}

Each coefficient in this construction is individually easy to compute: its binary representation is a $1$ followed by $k(n)$ zeros.
The noncomputability occurs in the limiting growth rate. \revmerged{Together with Theorem~\ref{thm:main}\ref{item:rational}, this places noncomputable exponential growth on the $O(n\log n)$ side of the same one-tape threshold.}
\section*{Acknowledgments}
The author was funded by the Czech Science Foundation project 24-11820S (``Automatic Combinatorialist'').

\section*{Acknowledgment of AI Use}
This note arose from a broader exploratory research project in which the author developed ideas and directions in dialogue with AI tools, including Sol, Astra, and Fable. The manuscript was then developed iteratively, paragraph by paragraph, with these tools assisting with exposition and critique. The author independently verified the mathematical claims and cited sources, made all final decisions, and takes full responsibility for the paper.

\bibliographystyle{abbrvnat}
\bibliography{accepting_path_one_tape}
\end{document}